 \documentclass[draft]{article}

\usepackage{amsmath,amsfonts,amsthm,amssymb,amscd,color}
\renewcommand{\Re}{\mathop{\rm Re}\nolimits}

\theoremstyle{plain} \newtheorem{theorem}{Theorem}[section]
\newtheorem{lemma}[theorem]{Lemma}
\newtheorem{proposition}[theorem]{Proposition}
\newtheorem{corollary}[theorem]{Corollary} \theoremstyle{definition}
\newtheorem{definition}[theorem]{Definition} \theoremstyle{remark}
\newtheorem{remark}[theorem]{Remark} \newtheorem{example}[theorem]{Example}

\newcommand{\R}{{\mathbb R}}

\newcommand{\Z}{{\mathbb Z}}

\newcommand{\N}{{\mathbb N}}

 \def\im{{\rm i}}

\newcommand{\C}{\mathbb{C}}

\newcommand{\bsig}{{\boldsymbol \sigma}}

\def\({\left(}
\def\){\right)}
\def\<{\left\langle}
\def\>{\right\rangle}

\newcommand{\zd}{{\delta\mathbb Z}}
\newcommand{\hd}{{\mathcal H_\delta}}
\newcommand{\sd}{{\mathcal S_\delta}}
\newcommand{\cd}{{\mathcal  C_\delta }}
\newcommand{\nd}{{\mathcal  N_\delta }}
\newcommand{\ud}{{\mathcal U_\delta }}

\newcommand{\hchd}{{\hat{\mathcal H}_\delta }}

\newcommand{\id}{{\mathfrak I_\delta }}
\newcommand{\hid}{{\hat{\mathfrak I}_\delta }}
\newcommand{\udi}{{U_{\mathrm{Dirac}} }}

\numberwithin{equation}{section}

\begin{document}

\title{Continuous limits of linear and nonlinear quantum walks}

\author {Masaya Maeda, Akito Suzuki}

\maketitle

\begin{abstract}
In this paper, we consider the continuous limit of a nonlinear quantum walk (NLQW)
that incorporates a linear quantum walk as a special case.
In particular, we rigorously prove that the walker (solution) of the NLQW on a lattice $\delta \Z$ uniformly converges (in Sobolev space $H^s$) to the solution to a nonlinear Dirac equation (NLD) on a fixed time interval as $\delta\to 0$.
Here, to compare the walker defined on $\delta\Z$ and the solution to the NLD defined on $\R$, we use Shannon interpolation.
\end{abstract}

\section{Introduction}
Quantum walks (QWs), or more precisely discrete time QWs, are quantum counterparts of classical random walks \cite{ADZ93PRA, ABNVW01, Gudder88Book}.  We can find the early prototypes in the context of Feynman path integral \cite{FH10Book, Riazanov58SPJETP} and Quantum Lattice Boltzmann methods \cite{Bialynicki-Birula94PRD, Meyer96JSP, SB93PD}.
QWs are now attracting diverse interest because of its connection to various fields of mathematics and physics such as orthogonal polynomials on the unit circle \cite{CGMV10CPAM}, quantum search algorithms \cite{AKR05Proc},
topological insulators \cite{
Kitagawa12QIP, KRBD10PRA}.
Further, since the early works studying QWs \cite{Bialynicki-Birula94PRD, FH10Book, Gudder88Book,  Riazanov58SPJETP,  SB93PD} are all more or less motivated by the discretization of Dirac equation, the relation between QWs and Dirac equation and other wave equation have repeatedly discussed by many authors from various viewpoints \cite{ANF14JPA, BDT15AP,BES07PRA, LKN15PRA,MB18PRA, MBD13PRA, MD12JMP, Sh13,Strauch06PRA73}.

Nonlinear QWs (NLQWs), which are nonlinear versions of the usual (linear) QWs with the nonlinearity coming into the dynamics through the state-dependence of the quantum coin, was first proposed in \cite{NPR07PRA} as an ``optical Galton board \cite{BMKSW99PRA}" with Kerr effect.
From then, several models of NLQWs have been proposed motivated by simulating nonlinear Dirac equations (NLD) \cite{LKN15PRA,MDB15PRE} and studying the nonlinear effect to the topological insulators \cite{GTB16PRA}.
See also \cite{MSSSS18DCDS, MSSSS18QIP, MSSSSnum} for the study of scattering phenomena, weak limit theorem and soliton-like behavior for NLQWs.
We note that for continuous time QWs, which are substantially described by discrete Schr\"odinger equations, nonlinear models are also attracting interest because it can speed up the quantum search \cite{MW13NJP}.

In this paper, motivated by the above works, we study the connection between NLQWs and nonlinear Dirac equations (NLD).
In particular, we show that the walker (or the solution) of NLQWs converges to the solution to the NLD.
Very roughly, we show that for fixed $T>0$,
\begin{align}\label{roughmain}
\| u(m \delta) -v_\delta(m \delta)\|_{L^2}\to 0\ \text{as}\ \delta\to0,\quad \text{uniformly for }m\in\Z,\ 0\leq m\leq T/\delta,
\end{align}
where $u$ is the solution to the NLD, $v_\delta$ is the walker of NLQW on $\delta\Z$(for the precise statement, see Theorem \ref{thm:1}
below).
Thus, we see that the walker converges to the solution to the NLD uniformly in a fixed time interval.
We emphasize that our model incorporates a linear quantum walk as a particular case and hence Theorem \ref{thm:1} says that the walker of the linear QW also converges to the solution to a Dirac equation. 

This paper is organized as follows. 
In Subsections \ref{subsec:NLQW} and \ref{subsec:NLD} we introduce NLQW and NLD respectively.
In Subsection \ref{subsec:Main}, we state our main result Theorem \ref{thm:1}.
In Section \ref{sec:pre} we recall some facts of the Shannon interpolation.
In Section \ref{sec:Proof}, we prove Theorem \ref{thm:1}.

\subsection{Nonlinear quantum walks}\label{subsec:NLQW}
We now introduce NLQWs, which are space-time discretized dynamics conserving $l^2$ norm (in the linear case, it is a unitary dynamics).
Let $\delta>0$ be a constant and set
\begin{align}\label{zd}
\zd:=\{\delta n\ |\ n\in \Z\}.
\end{align}
We set
\begin{align}\label{def:hd}
\hd :=l^2(\delta\Z,\C^2),\quad \<u,v\>_{ \hd}:=\delta \sum_{x\in\zd}\<u(x),v(x)\>_{\C^2}\text{ and }\|u\|_{ \hd}^2:=\<u,u\>_{ \hd},
\end{align}
where $\<\cdot,\cdot\>_{\C^2}$ is the inner-product
of $\C^2$, i.e.\ for $u=(u_1,u_2)\in \C^2$ and $v=(v_1,v_2)\in \C^2$, $\<u,v\>_{\C^2}= \sum_{j=1,2}u_j\bar v_j$, and we set $\|u\|_{\C^2}^2:=\<u,u\>_{\C^2}$.
We also use the standard Pauli matrices
\begin{align}\label{pauli}
\sigma_0:=\begin{pmatrix}1&0\\0&1\end{pmatrix},\  \sigma_1:=\begin{pmatrix}0 &1\\1&0\end{pmatrix},\  \sigma_2:=\begin{pmatrix}0&-\im \\ \im &0\end{pmatrix},\  \sigma_3:=\begin{pmatrix}1 &0\\0&-1\end{pmatrix}.
\end{align}
We define a shift operator $\sd:\hd\to\hd$ by
\begin{align}\label{sd}
\sd:=\begin{pmatrix} \mathcal T_{+,\delta}& 0\\ 0 & \mathcal T_{-,\delta} \end{pmatrix},\quad \mathcal T_{\pm, \delta}f:=f(\cdot\mp \delta).
\end{align}
It is clear that $\sd$ is a unitary operator on $\hd$.
\begin{remark}
Recall (formally) $\(e^{t \partial_x}f\)(x)=f(x+t)$.
Thus, we can express $\sd$ as
\begin{align}\label{sd2}
\sd=\begin{pmatrix} e^{-\delta \partial_x}& 0\\ 0 & e^{\delta \partial_x} \end{pmatrix}=e^{-\delta \sigma_3 \partial_x}=e^{\im \delta 	 A},\quad A:=\im \sigma_3 \partial_x.
\end{align}
In the following, we will use the expression of \eqref{sd2} for $\sd$.
\end{remark}

We next fix a smooth function $\mathbf s=(s_0,s_1,s_2,s_3):\R\to \R^4$ and set a linear coin operator 
$
\mathcal C_{\delta,\mathbf s}:\hd\to\hd$ by
\begin{align*}
\( 
\mathcal C_{\delta,\mathbf s}  u\)(x):=e^{-\im \delta \mathbf s(x)\cdot \bsig}u(x),\quad x\in \zd,
\end{align*}
where $\mathbf s(x)\cdot \bsig=\sum_{\alpha=0}^3s_\alpha(x)\sigma_\alpha$.
For simplicity, we often suppress the explicit dependence on 
$\mathbf s$ and write $\cd$ for $\mathcal C_{\delta,\mathbf s}$.
Since for each $x\in \zd$, $e^{-\im \delta \mathbf s(x)\cdot \bsig}$ is a $2\times 2$ unitary matrix, it is clear that $  \cd $ is unitary operator on $\hd$.


To define a nonlinear (or state-dependent) coin operator,
we fix $\gamma $ to be a $2\times 2$ Hermitian matrix and smooth function $g\in C^\infty(\R,\R)$.
We now define the nonlinear coin 
$
\mathcal N_{\delta,\gamma,g }:\hd\to\hd$ by
\begin{align}\label{nonlinearity}
\(
\mathcal N_{\delta,\gamma,g }
u\)(x) =e^{-\im\delta g(\<u(x),\gamma  u(x)\>_{\C^2})\gamma }u(x), 
\quad x\in \zd.
\end{align}
When there is no ambiguity we drop the dependence on $\gamma$ and $g$ and write just $\nd$ for $\mathcal N_{\delta,\gamma,g }$. 
Notice that since
$
\<\nd u(x), \nd u(x)\>_{\C^2}=\<u(x),u(x)\>_{\C^2},
$
we have
\begin{align}\label{nunit}
\|\nd u\|_{\hd}=\|u\|_{\hd}.
\end{align}


\begin{definition}
For $u_0\in  \hd$ and $m\in \Z$, $m\geq 0$, we define $\ud(m)u_0\in  \hd$ by the recurrence relation 
\begin{align}\label{qw1}
\ud(0)u_0=u_0,\quad
\ud(m+1)u_0=\mathcal S_\delta	 \cd  \nd\(\ud(m)u_0\).
\end{align}
\end{definition}

\begin{remark}
If $g=0$, then $\mathcal U_\delta$ is a linear unitary operator.
However, if $g\neq 0$, $\mathcal U_\delta$ becomes a nonlinear operator.
This is the reason why we need to define $\mathcal U_\delta(t)u_0$ be the recurrence relation \eqref{qw1}.
\end{remark}

We give several examples of our model, which cover various QWs appeared in the literature.

\begin{example}[Free QWs]
When $g=0$ and $\mathbf s$ do not depend on $x\in \Z$, we will call the corresponding QW a free QW. 
This quantum walk is also called homogeneous since the coin operator is spatially homogeneous.
A typical example is the case $\mathbf s=(0,-1,0,0)$, which appeared in the Feynman checkerboard model \cite{FH10Book}.
In particular, the coin operator in this case have the form
\begin{align*}
\mathcal C_{\delta,(0,-1,0,0)} =e^{\im \delta\sigma_1}=\begin{pmatrix} \cos \delta& \im \sin \delta\\ \im \sin \delta & \cos \delta\end{pmatrix}.
\end{align*}
Another important example is the Hadamard walk \cite{ABNVW01}, which is usually considered for the case $\delta=1$, and the coin operator is given by the Hadamard matrix:
\begin{align}\label{Hadamardcoin}
\mathcal C_{1,\frac{\pi}{4}(2,0,1,-2)} =e^{-\im \frac{\pi}{4}\(2,0,1,-2\)\cdot \bsig}=\frac{1}{\sqrt{2}}\begin{pmatrix} 1&1\\ 1 & -1\end{pmatrix}.
\end{align}
\end{example}

\begin{example}[Linear QWs]
When $g=0$, we will call the corresponding QW a linear QW.
A typical example will be the case $\mathbf s(x)=(0,0,\theta(x),0)$ where $\theta:\R\to \R$ is a function converging to some limit $\theta_\pm$ as $x\to \pm \infty$.
In this case, the coin operator
\begin{align*}
\mathcal C_1=R(\theta)=\begin{pmatrix} \cos \theta & -\sin \theta \\ \sin \theta & \cos \theta\end{pmatrix}
\end{align*}
is spatially inhomogeneous and called a position-dependent coin.
Such a model appears in the context of topological insulators \cite{Kitagawa12QIP} and the scattering theory for linear QWs are studied in \cite{MSSSSdis, MoriokaQW1, RST18LMPII,RST18LMPI, Suzuki16QIP}.
\end{example}

\begin{example}[NLQWs, I]\label{NPR}
When $g\neq 0$, we will call the corresponding QW a NLQW.
NLQW first proposed in \cite{NPR07PRA} is of the form
\begin{align*}
\mathcal U_{\mathrm{NPR}}:=\mathcal S_1 \mathcal C_{1,\frac{\pi}{4}(2,0,1,-2)} \mathcal N_{1,\gamma_1,g} \mathcal N_{1,\gamma_2,g},
\end{align*}
where the linear coin is given by the Hadamard matrix \eqref{Hadamardcoin}
and the two nonlinear coins are defined by $g(s)=\lambda s$ ($\lambda\in\R$), $\gamma_1=\frac{1}{2}(\sigma_0+\sigma_3)$ and $\gamma_2=\frac{1}{2}(\sigma_0-\sigma_3)$.
In particular, for $u=(u_1,u_2)$, we have
\begin{align*}
\mathcal N_{1,\gamma_1,g} u(x)=\begin{pmatrix}e^{\im \lambda |u_1(x)|^2}&0\\ 0 & 1\end{pmatrix}u(x),\quad \mathcal N_{1,\gamma_2,g} u(x)=\begin{pmatrix}1&0\\ 0 & e^{\im \lambda |u_2(x)|^2}\end{pmatrix}u(x).
\end{align*}

\end{example}

\begin{remark}
Even though our result in this paper is proved for 
the case of a single nonlinear coin, it can be extended to 
the case of two nonlinear coins as stated above without difficulty.
\end{remark}

\begin{example}[NLQWs, II]\label{GNT}
Another example of NLQW, which was proposed in \cite{LKN15PRA} as a simulator of a nonlinear Dirac equation, is
\begin{align*}
\mathcal U_{LKN}:=\mathcal S_\delta\mathcal C_\delta \mathcal N_{\delta,\sigma_j,g} ,
\end{align*}
where $g(s)= s$, $\mathcal C_1=R(\theta)$ ($\theta\in\R$) and $j=0$ or $3$.
The case $j=3$ is for simulating the Gross-Neveu model (scaler type interaction)
and
the nonlinear coin is of the form
\begin{align*}
\mathcal N_{\delta,\sigma_3,g}u(x)=\begin{pmatrix} e^{-\im \delta (|u_1(x)|^2-|u_2(x)|^2)} & 0 \\ 0 & e^{\im \delta (|u_1(x)|^2-|u_2(x)|^2)}\end{pmatrix}.
\end{align*}
The case $j=0$ is for simulating the Thirring model (vector type interaction) and
the nonlinear coin is of the form
\begin{align*}
\mathcal N_{\delta,\sigma_3,g}u(x)=e^{-\im \delta (|u_1(x)|^2+|u_2(x)|^2)}u(x).
\end{align*}
\end{example}

\subsection{Nonlinear Dirac equations in $1+1$ space-time}\label{subsec:NLD}
The Dirac equation on $\R$ is given by
\begin{align}\label{Dirac}
\im \partial_t u=-\im \sigma_3\partial_xu +\mathbf s \cdot \bsig u+g(\<u,\gamma u\>_{\C^2})\gamma  u,\quad (t,x)\in \R\times \R,\ u:\R\to\C^2.
\end{align}
Here, $\mathbf s:\R\to\R^4$, $\gamma$ is a $2\times2$ Hermitian matrix and $g:\R\to\R$ corresponds to the ones given in the definition of the NLQW.
Indeed, we will show that a solution to the NLQW converges to a solution to the NLD with the same $\mathbf s$, $\gamma$ and $g$.
We will denote the solution 
$u = u(t)$
to the Dirac equation \eqref{Dirac}
with the initial condition $u(0) = u_0$
by
\begin{align*}
u(t)=U_{\mathrm{Dirac}}(t)u_0. 
\end{align*}
Note that if the Dirac equation \eqref{Dirac} is
nonlinear  (i.e.\ if $g$ is not a constant), 
then so is $U_{\mathrm{Dirac}}$.
For a comprehensive introduction for the linear Dirac equation, see \cite{Thaller92Book}.

As the NLQW, we introduce several examples of the NLD.
\begin{example}[NLD: Gross-Neveu model and Thirring model]
For the case $\mathbf s=(0,0,m,0)$ ($m$ is a constant),  $\gamma=\sigma_3$ (resp. $\gamma=\sigma_0$), $g(s)=s$, NLD \eqref{Dirac} is called the Gross-Neveu model \cite{GN74PRD} (resp.\ Thirring model \cite{Thirring58AP}).
Further, a generalized Gross-Neveu model, which is the case of general $g\in C^\infty(\R,\R)$, has been studied in \cite{CPS17AIHPAN}.
\end{example}

\begin{example}[Nonlinear coupled mode equations]\label{ex:NCME}
Let $\mathbf s=(V,\kappa,0,0)$ and suppose that
$V,\kappa\in C^\infty(\R,\R)$ are bounded functions and $\gamma_1=\frac12 (\sigma_0 +\sigma_3)$, $\gamma_2=\frac12 (\sigma_0 - \sigma_3)$. 
Then the NLD becomes
\begin{align*}
\im \partial_t u = -\im \sigma_3 u + \mathbf s\cdot \bsig u + 2\<u,u\>_{\C^2}u -\<u,\gamma_1u\>_{\C^2}\gamma_1u-\<u,\gamma_2u\>_{\C^2}\gamma_2u.
\end{align*}
Such a model appears in the study of nonlinear propagation of light in an optical fiber waveguide \cite{GSWK05CM, GWH01JNS}.
A similar model also appears in the study of Bose-Einstein condensation.
In particular, in \cite{PCP06PRE}, the following model is studied:
\begin{align*}
\im \partial_t u = -\im \sigma_3 u + \mathbf s\cdot \bsig u + \<u,u\>_{\C^2}^2 u -2\<u,\gamma_1u\>_{\C^2}^2\gamma_1u-2\<u,\gamma_2u\>_{\C^2}^2\gamma_2u.
\end{align*}
As we remarked in Example \ref{NPR}, our result in this paper can be generalized to the case of several nonlinear coins without difficulty.
\end{example}

We introduce some mathematical results about NLD.
To do so, we prepare several notations.

We set $L^2=L^2(\R,\C^2)$ and $H^s:=H^s(\R,\C^2)$ ($s\in \N$), the $\C^2$-valued Sobolev spaces.
The inner product of $L^2$ will be denoted by
\begin{align*}
\<u,v\>:=\int_\R \<u(x),v(x)\>_{\C^2}\,dx.
\end{align*}
We set $\|u\|_{L^2}:=\<u,u\>^{1/2}$.
The norm of $H^s$ is defined by
\begin{align}\label{norm:Hs}
\|u\|_{H^s}^2:=\sum_{j=0}^s
\| \partial_x^j u\|_{L^2}^2.
\end{align}
We further, define the innerproduct of $H^s$ by
\begin{align*}
\<u,v\>_{H^s}:=\sum_{j=0}^s \<\partial_x^j u, \partial_x^j v\>.
\end{align*}

Since for $s\geq 1$, $H^s$ becomes an algebra, one can show the following result by standard fixed point argument.
\begin{proposition}\label{lwp:dirac}
Let $s\geq 1$ and suppose that
$\|\mathbf s\|_{L^\infty}+\|\mathbf s'\|_{H^{s-1}}<\infty$.
Let $L>0$. Then there exists $T>0$ such that a unique solution $u(t) = U_{\mathrm{Dirac}}(t)u_0\in C([0,T],H^s )$ of NLD \eqref{Dirac} exists for any  $u_0\in H^s$ with $\|u_0\|_{H^s}\leq L$.
Further, for $u_j \in H^s$ with $\|u_j\|_{H^s} \leq L$ ($j=1,2$), 
$U_{\mathrm{Dirac}}(t)u_j$ satisfies
\begin{align*}
\sup_{t\in[0,T]}\|U_{\mathrm{Dirac}}(t)u_1 - U_{\mathrm{Dirac}}(t) u_2 \|_{H^s}\leq C \|u_1-u_2\|_{H^s},
\end{align*}
where $C$ is a constant depends only on $L$.
\end{proposition}

Inspired by the above result for the solutions of nonlinear Dirac equations, we define the  condition $(\mathrm{Lip})_s$ as follows.
\begin{definition}\label{def:unifLip}
Let $s\geq 1$ and $T,L>0$.
We say that the pair $(T,L)$ satisfies condition $(\mathrm{Lip})_s$
if there exists a constant $C_{T,L}>0$ such that for $u_j\in  H^s$ 
with $\|u_j\|_{ H^s}\leq L$ ($j=1,2$), 
$\udi(\cdot)u_j\in C([0,T],H^s)$ and 
\begin{align}\label{unifLip}
\sup_{0\leq t\leq T}\|\udi(t) u_1-\udi(t) u_2\|_{  H^s}\leq C_{T,L} \|u_1-u_2\|_{ H^s}.
\end{align}
\end{definition}

\begin{remark}
By proposition \ref{lwp:dirac}, for any $L>0$, we can always find $T>0$ such that $(T,L)$ satisfies $(\mathrm{Lip})_s$.
Moreover, if NLD \eqref{Dirac} is globally wellposed,
we can take $T=\infty$ in proposition \ref{lwp:dirac} for arbitrary $L>0$
and hence
$(\mathrm{Lip})_s$ holds for arbitrary $(T,L)\in \R_+\times \R_+$.
It is known that if the nonlinearity 
comes only from $\sigma_0$, $\gamma_1$ and  $\gamma_2$ (as Example \ref{ex:NCME}), then the NLD is globally wellposed
(see \cite{Pelinovsky11RIMS} for more information).
\end{remark}

\subsection{Main results}\label{subsec:Main}

To state our result precisely, we introduce some notation.
When there exists a constant $C>0$ 
such that $a\leq Cb$, we write $a\lesssim b$ or $b\gtrsim a$.
If the implicit constant $C$ depends on some parameter $\alpha$ ($C=C_\alpha$), then we write $a\lesssim_\alpha b$.
If $a\lesssim b$ and $b\lesssim a$, we write $a\sim b$.

For Banach spaces $X,Y$, we set $\mathcal L(X,Y)$ to be the Banach space of all bounded linear operators from $X$ to $Y$.
We set $\mathcal L(X):=\mathcal L(X,X)$.

We set $\hchd:=L^2(\R/2\pi \delta^{-1}\Z,\C^2)$
 and define the inner product and norm by 
\begin{align*}
\<u,v\>_{\hchd}:=\int_{-\pi/\delta}^{\pi/\delta}\<u(\xi),v(\xi)\>_{\C^2}\,d\xi,\quad
 \|u\|_{\hchd}^2:=\<u,u\>_{\hchd }.
\end{align*}
We define the discrete Fourier transform $\mathcal F_\delta\in \mathcal L( \hd,\hchd )$ and its inverse $\mathcal F_\delta^{-1}\in \mathcal L(\hchd , \hd)$ as
\begin{align*}
\mathcal F_\delta u(\xi):=\frac{\delta}{\sqrt{2\pi}}\sum_{x\in \delta \Z} e^{-\im x \xi}u(x),\quad \mathcal F_\delta^{-1}v(x)=\frac{1}{\sqrt{2\pi}}\int_{-\pi/\delta}^{\pi/\delta} e^{\im x \xi}v(\xi)\,d\xi,
\end{align*}
where  $\hd$ is defined in \eqref{def:hd}.
Then, we have 
\begin{align*}
\<\mathcal F_\delta u, \mathcal F_\delta v\>_{\hchd }=\<u,v\>_{ \hd},\quad \<\mathcal F_\delta^{-1} u, \mathcal F_\delta^{-1} v\>_{ \hd}=\<u,v\>_{\hchd }.
\end{align*}

We denote the Fourier transform on $L^2$ by $F$.
In particular, we set
\begin{align*}
F u(\xi):=\frac{1}{\sqrt{2\pi}}\int_\R e^{-\im x \xi}u(x)\,dx,\quad F^{-1}u(x)=\frac{1}{\sqrt{2\pi}}\int_\R e^{\im x \xi}u(x)\,dx.
\end{align*}
We define $H_\delta\subset L^2$ by
\begin{align*}
H_\delta:=\{u\in L^2\ |\ \mathrm{supp} Fu\subset [-\pi/\delta,\pi/\delta]\}.
\end{align*}
Since the Fourier transform has compact support, we have $H_\delta\subset \cap_{s\geq 0}H^s$.
We define the projection to $H_\delta$ by
\begin{align}\label{def:j}
j_\delta= F^{-1} \chi_{[-\pi/\delta,\pi/\delta]}  F\in \mathcal L(L^2),
\end{align}
where $\chi_A$ is the characteristic function of $A$.
Obviously, we have
\begin{align}\label{j1}
j_\delta^2=j_\delta,\quad \mathrm{Ran}j_\delta= H_\delta,\quad \<j_\delta u,v\>=\<u,j_\delta v\>.
\end{align}

We next define  the Shannon interpolation (see, e.g.\  \cite{BF18pre}), which is an isometry from $\hd$ to $H_\delta\subset L^2$, by

\begin{align}\label{def:i}
\id:=F^{-1}\circ  \hid \circ \mathcal F_\delta \in \mathcal L( \hd,L^2),
\end{align}
where $\hid$ is the natural identification between $\hat{\mathcal H}_\delta$ and $H_\delta$ given by
\begin{align}\label{def:idelta}
\hid:\hchd \to H_\delta \subset L^2,\quad  \hid u (\xi):=\begin{cases} u(\xi),& \xi\in [-\pi/\delta,\pi/\delta]\\ 0, & |\xi|>\pi/\delta.
\end{cases}
\end{align}

We are now in a position to state our main result precisely.
\begin{theorem}\label{thm:1}
Let 
$s\geq 1$, $T>0$ and $L>0$.
Assume $\|\mathbf s\|_{L^\infty}+\| \mathbf s'\|_{H^{s}(\R,\R^4)}<\infty$ and $(T,L)$ satisfies condition $(\mathrm{Lip})_s$.
Assume $\|u_0\|_{ H^{s+1}}\leq L$.
Then, there exists $\delta_0>0$ 
such that for $\delta\in (0,\delta_0]$, 
\begin{align}\label{errorest}
\sup_{m\in\Z, 0\leq \delta m\leq T}\| \id \circ \ud(m)\circ \id^{-1}\circ j_\delta u_0 -\udi(m \delta)  u_0 \|_{H^s}
\lesssim_{T,L} \delta,
\end{align}
where the implicit constant is independent of $\delta$. 
\end{theorem}

Since $\|\cdot\|_{L^2}\leq \|\cdot\|_{H^s}$ for $s\geq 0$, 
we have the following continuous limit. 
\begin{corollary}
Under the same assumptions as in Theorem \ref{thm:1}, 
the walker 
of the NLQW
converges to the solution 
to the NLD in the following sense:
\[ 
\lim_{\delta \to 0} 
	\sup_{m\in\Z, 0\leq \delta m\leq T}\| \id \circ \ud(m)\circ \id^{-1}\circ j_\delta u_0 -\udi(m \delta)  u_0 \|_{L^2} = 0. 
\]
\end{corollary}

\begin{remark}
Theorem \ref{thm:1} calls for some explanation. 
For the solution $u(t) = \udi(t)  u_0 $ to the NLD in $\R$ with the initial condition 
$u(0) = u_0 \in H^{s+1}$ (on $\R$), 
we discretize the initial condition $u_0$ by $\id^{-1} \circ j_\delta$,   evolve it by the NLQW by \eqref{qw1}.
After m steps of the NLQW evolution, we can put it back to a function  by the Shannon interpolation $\id$ since it is unitary from 
$\hd$ to $H_\delta$ (see Lemma \ref{lem:shanon1}). 
Theorem \ref{thm:1} ensures that the resulting function $\id \circ \ud(m)\circ \id^{-1}\circ j_\delta u_0 $ successively approximates the solution to the NLD. In this sense, we can say that the continuous limit
of the NLQW is the NLD.
\end{remark}

Although many works discuss the continuous limit of (linear and nonlinear) QWs, it seems that they just informally compare the equation of QWs and the Dirac equation by, for instance, expanding the equation or referring the Trotter-Kato formula.
What is really needed is the estimate of the difference of the walker of the QW and solution to the Dirac equation as given in \eqref{errorest}.
The only result of such kind we are aware is \cite{ANF14JPA}, where the authors show \eqref{errorest} for $m=1$.
In this sense, our result is new even in the linear QWs.

From the viewpoint of numerical analysis, 
the NLQW 
gives a splitting method of the NLD.
Splitting methods are now popular numerical schemes for approximating semilinear Hamiltonian partial differential equations such as nonlinear Schr\"odinger equations \cite{Faou12Book}.
In this point of view, it may be interesting to investigate higher-order methods such as Strang splitting for the NLD.
However, to make our paper simple, we will not investigate them.

For the proof of Theorem \ref{thm:1}, we employ the energy method of Holden-Karlsen-Risebro-Tao \cite{HKRT11MC}, which was originally applied to the KdV equation.

%

\section{Preliminary}\label{sec:pre}

In this section, we collect technical tools which we use in the proof of Theorem \ref{thm:1}. Recall the definitions of $\<\cdot,\cdot\>_\hd$ and $\id$ given in \eqref{def:hd} and \eqref{def:i} respectively.

\begin{lemma}[\cite{BF18pre}]\label{lem:shanon1}
$\id:\hd \to H_\delta $ is unitary and 
\begin{align}\label{i1}
\id u (x)=u(x), \quad
x\in \delta\Z.
\end{align}
\end{lemma}

\begin{proof}
By definition, $\id$ is an isometry and the image of $\id$ is $H_\delta$. Hence, $\id$ is unitary. 
\eqref{i1} is shown by
\begin{align*}
\id u(x)=\frac{1}{\sqrt{2\pi}}\int_\R e^{\im x \xi}  \hid \mathcal F_\delta u(\xi)\,d\xi=\frac{1}{\sqrt{2\pi}}\int_{-\pi/\delta}^{\pi/\delta} e^{\im x \xi} \mathcal F_\delta u(\xi)\,d\xi=\mathcal F_\delta^{-1}\circ \mathcal F_\delta u (x)=u(x).
\end{align*}
\end{proof}

Recall $j_\delta$ given in \eqref{def:j}.
\begin{lemma}\label{lem:j2}
For $\sigma \geq 1$,
\begin{align}\label{j2}
\|(1- j_\delta) u\|_{H^{s}}\lesssim \delta^{\sigma}\|u\|_{H^{s+\sigma}}.
\end{align}
\end{lemma}
\begin{proof}
By the definition of $j_\delta$ and the norm of $H^s$ given in \eqref{norm:Hs}, we have
\begin{align*}
\|(1- j_\delta)u\|_{H^{s}}=
\| \<\cdot\>^{s} 
\chi_{\{|\cdot|>\pi/\delta\}}
\mathcal F u\|_{L^2}
\lesssim \delta^{\sigma}\|\<\cdot\>^{s+\sigma} \mathcal F u\|_{\mathcal H}=\delta^{\sigma}\|u\|_{\mathcal H^{s+\sigma}},
\end{align*}
where we have used the fact that $(\delta \<\xi\>)^\sigma \geq 1$ if $|\xi| > \pi/\delta$. 
Therefore, we have \eqref{j2}.
\end{proof}

Recall $\mathcal T_{-,\delta}u(x)=u(x+\delta)$.
We set 
\begin{align}\label{differenceop}
\mathcal D_\delta:=\delta^{-1}\(\mathcal T_{-,\delta}-1\)\in \mathcal L( \hd).
\end{align}
Formally, we can write $\mathcal D_\delta=\frac{e^{\delta \partial_x}-1}{\delta}$.
Further, we set $$ D_\delta:=\frac{e^{\delta \partial_x}-1}{\delta}\in \mathcal L(L^2).$$
\begin{remark}
The operators $\mathcal D_\delta$ and $D_\delta$ formally have the same definition.
However, $\mathcal D_\delta$ is defined on $\mathcal H_\delta=l^2(\delta\Z,\C^2)$ and $D_\delta$ is defined on $L^2(\R,\C^2)$.
\end{remark}

\begin{lemma}\label{lem:difdif}
Let $s\geq 0$.
For $u \in H^{1}$, we have 
\begin{align*}
\|D_\delta u\|_{L^2}\leq \|\partial_x u\|_{L^2}.
\end{align*}
\end{lemma}

\begin{proof}
By $e^{\delta \partial_x}=1+\delta  \int_0^1 e^{\delta t \partial_x }\,dt \partial_x$, we have
\begin{align*}
\|D_\delta u\|_{L^2}\leq \|\int_0^1 e^{\delta t \partial_x }\,dt\|_{\mathcal L(L^2)}\|\partial_x u\|_{L^2}.
\end{align*}
Since
\begin{align*}
\|\int_0^1 e^{\delta t \partial_x }v\,dt\|_{L^2}
=\|\int_0^1 v(\cdot+\delta t)\,dt\|_{L^2}\leq \|v\|_{L^2},
\end{align*}
we have $\|\int_0^1 e^{\delta t \partial_x }\,dt\|_{\mathcal L(L^2)}\leq 1$.
Therefore, we have the conclusion.
\end{proof}

\begin{lemma}\label{est:HdNorm}
Let $s\geq 0$ and $u\in H^{s+1}$.
Then, for $0\leq j\leq s$, we have
\begin{align}\label{HdNorm}
\delta \sum_{x\in \Z^d}\|D_\delta^j u(x)\|_{\C^2}^2\leq \|\partial_x^j u\|_{L^2}^2+2\delta\|\partial_x^j u\|_{L^2}\|\partial_x^{j+1} u\|_{L^2}.
\end{align}
\end{lemma}

\begin{remark}
By Sobolev embedding, we have $ u \in H^1(\R,\C^2)\hookrightarrow
C^1(\R,\C^2)
$.
Therefore $u$ is defined pointwise and $ D_\delta^j u (x)$ has a meaning.
\end{remark}

\begin{proof}
We first prove the case $j=0$.
Fix $x\in \delta\Z$.
Set $F_x(t)=\|u(x+t)\|_{\C^2}^2$.
Then, since $\partial_t F_x=2\Re\<u(x+t),\partial_x u(x+t)\>_{\C^2}\in L^1(\R)$, we have
\begin{align*}
F_x(t)=F_x(0)+\int_0^t \partial_t F_x(s)\,ds.
\end{align*}
By the Fubini Theorem,
\begin{align*}
\delta\sum_{x\in \delta\Z}\|u(x)\|_{\C^2}^2 
&=\sum_{x\in\delta\Z}\int_0^\delta F_x(0)\,
dt
= \sum_{x\in \delta\Z}\int_0^\delta F_x(t)\,dt-\sum_{x\in\delta\Z}\int_0^\delta \int_0^t \partial_t F_x(s) \,ds dt \\
&= \|u\|_{L^2}^2-2 \Re\sum_{x\in \delta\Z}\int_0^{\delta}(\delta-s)\<u(x+s),\partial_x u(x+s)\>_{\C^2}ds\\
&\leq 
\|u\|_{L^2}^2
+2 \delta \int_\R \|u(x)\|_{\C^2}\|\partial_x u(x)\|_{\C^2}\,dx.
\end{align*}
Therefore, by Schwartz, we have the conclusion.

Next, for $j\geq 1$, assume that we have \eqref{HdNorm} for $j-1$.
Then, by Lemma \ref{lem:difdif}, since $\partial_x$ and $D_\delta$ commute, we have
\begin{align*}
\delta\sum_{x\in \Z^d}\|  D_\delta^j u(x)\|_{\C^2}^2&\leq \|\partial_x^{j-1} D_\delta u\|_{L^2}^2+2\delta\|\partial_x^{j-1} D_\delta u\|_{L^2}\|\partial_x^j D_\delta u\|_{L^2}\\&\leq \|\partial_x^j u\|_{L^2}^2+2\delta\|\partial_x^j u\|_{L^2}\|\partial_x^{j+1} u\|_{L^2}.
\end{align*}
Therefore, we have the conclusion.
\end{proof}

\begin{lemma}\label{lem:difdif2}
Let $u\in H_\delta$.
Then, we have
\begin{align*}
\|\partial_x^j u\|_{L^2}\sim \| D_\delta^j u\|_{L^2}.
\end{align*}
Here, the implicit constant is independent of $\delta$.
\end{lemma}

\begin{proof}
First, for $|\eta|\leq \pi$, we have
\begin{align*}
|\int_0^1e^{\im \eta t}\,dt|\sim 1.
\end{align*}
We have
\begin{align*}
\|D_\delta^j u\|_{L^2}^2&=\int_{-\pi/\delta}^{\pi/\delta}\left|\frac{e^{\im \delta \xi}-1}{\delta}\right|^{2j} \|\hat u(\xi)\|_{\C^2}^2\,d\xi=\int_{-\pi/\delta}^{\pi/\delta}|\xi \int_0^1 e^{\im \delta \xi t}\,dt|^{2j} \|\hat u(\xi)\|_{\C^2}^2\,d\xi\\
&\sim\int_{-\pi/\delta}^{\pi/\delta}|\xi|^{2j} \|\hat u(\xi)\|_{\C^2}^2\,d\xi=\|\partial_x^j u\|_{L^2}^2.
\end{align*}
Therefore, we have the conclusion.
\end{proof}

\begin{proposition}\label{prop:shanon}
Let $s\geq 0$, $\sigma\geq 1$.
Let $u\in H_\delta$ and $v\in H^{s+\sigma}$ with $u(x)=v(x)$ for all $x\in \delta\Z\subset \R$.
Then, we have
\begin{align}\label{differencepoint}
\|u- v\|_{H^s}\lesssim  \delta^\sigma \|v\|_{H^{s+\sigma}}.
\end{align}
Here, the implicit constant is independent of $u,v$ and $\delta$.
\end{proposition}

\begin{remark}
Notice that the right hand side of \eqref{differencepoint} 
does 
not depend on $u$.
\end{remark}

\begin{proof}
First, by Lemma \ref{lem:j2}, we have
\begin{align}\label{pr:sha1}
\|u-v\|_{H^s}\lesssim \|u-j_\delta v\|_{H^s}+\delta^\sigma\|v\|_{H^{s+\sigma}}.
\end{align}
For $0\leq j\leq s$, since $u-j_\delta v\in H_\delta$, 
by Lemma \ref{lem:difdif2} we have
\begin{align*}
\|\partial_x^j(u-j_\delta v)\|_{L^2}\sim \|D_\delta^j(u-j_\delta v)\|_{L^2}.
\end{align*}
By Lemmas \ref{lem:shanon1}, \ref{lem:j2} and \ref{est:HdNorm}, we have
\begin{align*}
&\|D_\delta^j(u-j_\delta v)\|_{L^2}^2=\|\mathfrak I_\delta^{-1}\circ D_\delta^j\( u-j_d v\)\|_{ \hd}^2=\delta\sum_{x\in \Z_d}\| D_\delta^j\( u-j_d v\)(x)\|_{\C^2}^2\\&=\delta\sum_{x\in \Z^d}\|  D_\delta^j(j_\delta-1)v(x)\|_{\C^2}^2
\lesssim \|(j_\delta-1)\partial_x^j v\|_{L^2}^2+\delta \|(j_\delta-1)\partial_x^jv\|_{L^2}
\|(j_\delta-1)\partial_x^{j+1} v\|_{L^2}
\\
&\lesssim \delta^{2\sigma} 
\|v\|_{H^{j+\sigma}}^2.
\end{align*}
Therefore, we have the conclusion.
\end{proof}

%
%

\section{Proof of Theorem \ref{thm:1}}\label{sec:Proof}

In this section, we prove Theorem \ref{thm:1}.
In the following, as claimed in Theorem \ref{thm:1}, we fix $T,L>0$ and $s\geq 1$ and assume
\begin{align}\label{Assumption,thm:1}
\|u_0\|_{H^{s+1}}\leq L,\ \|\mathbf s\|_{L^{\infty}(\R,\R^4)}+\|\mathbf s'\|_{H^s(\R,\R^4)}<\infty
\end{align}
and 
\begin{align}\label{Assumption,thm:1,2}
(T,L)\ \text{satisfies condition } (\mathrm{Lip})_s.
\end{align}
Since $\mathbf s$ is fixed, we will not denote the dependence of $\|\mathbf s\|_{L^{\infty}(\R,\R^4)}+\|\mathbf s'\|_{H^s(\R,\R^4)}$ in the implicit constant in the inequalities below.

We start with decomposing $\| \hid \circ \ud(m)\circ \hid^{-1}\circ j_\delta u_0 -\udi(m \delta)  u_0 \|_{ H^s}$ as
\begin{equation}\label{Tri1}
\begin{aligned}
&\| 
\hid 
\circ \ud(m) \circ \hid^{-1}\circ j_\delta u_0 -\udi(m \delta)  u_0\|_{ H^s}\\&\quad\leq 
\|
\hid 
\circ \ud(m) \circ \hid^{-1}\circ j_\delta u_0 -\udi(m \delta)  j_\delta u_0\|_{ H^s}+\|\udi (m \delta)  j_\delta u_0-\udi (m \delta)  u_0\|_{ H^s}.
\end{aligned}
\end{equation}
For $m \delta\leq T$, one can estimate the second term of the right hand side of \eqref{Tri1} by
the assumptions \eqref{Assumption,thm:1} and \eqref{Assumption,thm:1,2}.
Indeed, 
by Lemma \ref{lem:j2}, 
\begin{align}\label{pthm1:1}
\|\udi(m \delta)  j_\delta u_0-\udi(m \delta)  u_0\|_{H^s}\lesssim_{T,L} \|(j_\delta-1) u_0\|_{ H^s}\lesssim_{T, L} \delta .
\end{align}
We further decompose the first term of \eqref{Tri1} as
\begin{equation}
\begin{aligned}\label{pthm2:2}
&\|
\hid 
\circ \ud(m) \circ \hid^{-1}\circ j_\delta u_0 -\udi(m \delta)  j_\delta u_0\|_{ H^s}\\&\quad\leq
 \|
\hid 
 \circ \ud(m) \circ \hid^{-1}\circ j_\delta u_0 - U_\delta(m) j_\delta u_0\|_{ H^s}+\| U_\delta(m) j_\delta u_0 - \udi(m \delta)  j_\delta u_0\|_{ H^s},
\end{aligned}
\end{equation}
where
\begin{align}\label{def:qwincont}
U_\delta(0)u_0=u_0,\quad
U_\delta(m+1)u_0= S_\delta	 C_\delta  N_\delta\(U_\delta(m)u_0\).
\end{align}
and
\begin{align}\label{def:qwincont2}
S_\delta:=\begin{pmatrix} e^{-\delta \partial_x}& 0 \\ 0 & e^{\delta \partial_x} \end{pmatrix}, \ C_\delta:=e^{-\im \delta \mathbf s(\cdot)\cdot \bsig} \text{ and } N_\delta:=e^{-\im g(\<\cdot,\gamma \cdot\>_{\C^2})\gamma}\cdot.
\end{align}

\begin{remark}
$U_\delta$, $S_\delta$, $C_\delta$ and $N_\delta$ are the continuous counterparts of $\ud$, $\sd$, $\cd$ and $\nd$ respectively.
That is, $U_\delta$, $S_\delta$, $C_\delta$ and $N_\delta$ is defined on $L^2(\R,\C^2)$ instead of $\mathcal H_\delta$ with formally the same definition as $\ud$, $\sd$, $\cd$ and $\nd$.
\end{remark}

We next bound the second term in \eqref{pthm2:2} following Holden-Karlsen-Risebro-Tao \cite{HKRT11MC}.
To this end, we introduce $v_\delta(t_1,t_2,t_3)$ as follows. 
Let 
\begin{align}\label{def:omega1}
\Omega_\delta =\cup_{m\in \Z_{\geq 0}}\Omega_\delta^m,
\end{align}
where $\Omega_\delta^m:=[m \delta, (m+1)\delta]^3$. 
We define self-adjoint operators $A$ and $B$
as
\[  A = -\im \sigma_3 \partial_x,
	\quad B = \mathbf {s}\cdot \bsig. \]
We define a nonlinear operator $G$ as
\[ G(v) = 	g(\<v, \gamma  v \>_{\C^2})\gamma  v,
	\quad v \in L^2(\R,\C^2). \]
Let $v_\delta(0,0,0)= j_\delta u_0 \in L^2(\R,\C^2)$ and define
$v_\delta(t_1,t_2,t_3)\in L^2(\R,\C^2)$ for $(t_1,t_2,t_3)\in \Omega_\delta $ by
\begin{equation}\label{def:v}
\begin{aligned}
\im \partial_{t_1} v_\delta
&= G(v_\delta),\quad t_2=t_3=\delta m,\\
\im \partial_{t_2} v_\delta 
&= Bv_\delta,\quad t_3=\delta m,\\
\im \partial_{t_3} v_\delta 
& = Av_\delta.
\end{aligned}
\end{equation}
More precisely, given the value of $v_{\delta}(m \delta,m \delta, m \delta)$,  we are defining $v_{\delta}(\tilde t_1,\tilde t_2,\tilde t_3)$ for $(\tilde t_1,\tilde t_2,\tilde t_3)\in \Omega_\delta^m$ by first solving the first equation of \eqref{def:v} in the $t_1$ direction up to $t_1=\tilde t_1$ and then solve the second equation of  \eqref{def:v} in the $t_2$ direction up to $t_2=\tilde t_2$ and finally solve the third equation of \eqref{def:v} in the $t_3$ direction up to $t_3=\tilde t_3$.
By this procedure we can define $v_{\delta}((m+1)\delta,(m+1)\delta,(m+1)\delta)$ and thus we can define the value of $v_\delta$ for all $(t_1,t_2,t_3)\in \Omega_\delta$ by induction because $v_\delta(0,0,0)=j_\delta	 u_0$ is given.

\begin{remark}
We note that $v_\delta=v_\delta(t_1,t_2,t_3,x)$ is a $\C^2$-valued function defined on $\Omega_\delta \times \R$.
However, since we want to view $v_\delta$ as an $L^2(\R,\C^2)$-valued function on $\Omega_\delta$, we write $v_\delta=v_\delta(t_1,t_2,t_3)$ and suppress the dependence on the spatial variable $x$.
We further remark that the differential operator $A=-\im \sigma_3 \partial_x$ acts on this spatial variable $x$.
\end{remark}

\begin{lemma}\label{lem:vdeltaisU}
Let $v_\delta$ be the solution to \eqref{def:v} with $v_\delta(0,0,0)= j_\delta u_0$.
Then, $v_\delta$ correspond to $U_\delta(\cdot) j_\delta u_0$ at the diagonal lattice point.
That is, we have
\begin{align}\label{vdeltaisU}
v_\delta(\delta m, \delta m, \delta m)= U_\delta(m) j_\delta u_0.
\end{align}
\end{lemma}
\begin{proof}
Recall \eqref{def:qwincont} and \eqref{def:qwincont2}.
We prove \eqref{vdeltaisU} by induction.
Thus, we can assume \eqref{vdeltaisU}.
Our goal will be to show \eqref{vdeltaisU} with $m$ replaced by $m+1$.
We first show
\begin{align}\label{1ststep}
v_{\delta}((m+1)\delta,m \delta, m \delta)=N_\delta(U_\delta(m)j_\delta u_0).
\end{align}
By the first equation of \eqref{def:v}, 
\begin{align*}
\frac{d}{dt_1} \<v_{\delta}(t_1,\delta m, \delta m), \gamma v_{\delta}(t_1,\delta m, \delta m)\> 
&=\<-\im G(v_{\delta}), \gamma v_{\delta}\> -\<v_{\delta}, \im \gamma G(v_{\delta})\> \\&
=\<-\im g(\<v_{\delta}, \gamma  v_{\delta}\>_{\C^2})\gamma  v_{\delta}, \gamma v_{\delta}\> -\<v_{\delta}, \im g(\<v_{\delta}, \gamma  v_{\delta}\>_{\C^2}) \gamma^2  v_{\delta}\>=0.
\end{align*}
Hence, $\<v_{\delta}(t_1,\delta m, \delta m), \gamma v_{\delta}(t_1,\delta m, \delta m)\>$
conserves. 
By \eqref{def:qwincont2} and the first equation of \eqref{def:v} again, 
we obtain \eqref{1ststep}.
Similarly, from
the second and third equations of \eqref{def:v},
we can prove
\begin{align*}
v_\delta((m+1)\delta,(m+1)\delta,m \delta)=C_\delta N_\delta(U_\delta(m)j_\delta u_0),
\end{align*}
and
\begin{align*}
v_\delta((m+1)\delta,(m+1)\delta,(m+1) \delta)=S_\delta	C_\delta N_\delta(U_\delta(m)j_\delta u_0).
\end{align*}
Therefore, we have the conclusion.
\end{proof}
Setting $v_\delta(t):=v_\delta(t,t,t)$, we show the following proposition.
\begin{proposition}\label{prop:mainthm1}
For sufficiently small $\delta>0$,
we have 
\begin{align}\label{eq:mainthm1}
\sup_{t\in[0,T]}\|v_\delta(t)-u_\delta(t)\|_{ H^s}\lesssim_{T,L} \delta,
\end{align}
where $u_\delta(t):=\udi(t) j_\delta u_0$.
\end{proposition} 

\begin{remark}\label{rem:pr}
By Lemma \ref{lem:vdeltaisU} and Proposition \ref{prop:mainthm1}, we obviously have
\begin{align}\label{eq:mainthm12}
\|U_\delta(m) u_0-
\udi(m\delta) j_\delta 
u_0\|_{ H^s}\lesssim_{T,L} \delta,\quad \text{for }m\in\N,\ m \delta \leq T,
\end{align}
where the implicit constant are independent of $m,\delta$.
Thus, we obtain the bound for the second term of \eqref{pthm2:2}.
It remains to obtain the bound for the first term of \eqref{pthm2:2}.
\end{remark}

Before proving Proposition \ref{prop:mainthm1}, we prepare several notations and lemmas.
First, we set
\begin{align}\label{Ndash}
G '(v)w=2g'(\<v ,\gamma  v \>_{\C^2})\Re\<w,\gamma  v \>_{\C^2}\gamma  v  +g(\<v ,\gamma  v \>_{\C^2})\gamma  w,
\end{align}
where $G'(v)$ is the Fr\'echet derivative of $G$, $\Re \<w,\gamma v\>_{\C^2}$ is the real part of $\<w,\gamma v\>_{\C^2}$
and
\begin{align}\label{commutator}
 [X,G](v ):=X G(v )-G'(v )X v \text{\ for }X=A,B.
\end{align}

\begin{lemma}\label{lem:com}
We have
\begin{align*}
\|[A,B]v\|_{ H^s}\lesssim \|v\|_{ H^{s+1}},\quad
\|[A,G]v\|_{ H^s}\lesssim_{\|v\|_{H^s}}\|v\|_{ H^{s+1}},\quad 
\|[B,G]v\|_{ H^s}\lesssim_{\|v\|_{ H^s}}1.
\end{align*}
\end{lemma}

\begin{proof}
First, recall \eqref{Assumption,thm:1}.
Thus, by
\begin{align*}
[A,B]v=-\im \sigma_3 (\mathbf s'\cdot \bsig) v+\im \mathbf s \cdot [\bsig,\sigma_3]\partial_x v,
\end{align*}
the bound for $\|[A,B]v\|_{ H^s}$ is obvious since for $s\geq 1$, $H^s$ becomes an algebra.
Next, we have
\begin{align*}
[A,G](v)=&\im g(\< v, \gamma v\>_{\C^2}) [\gamma,\sigma_3] \partial_x v+2 g'(\< v, \gamma v\>_{\C^2})\(\Re\<\im \sigma_3 \partial_x v, \gamma v\>_{\C^2}\gamma v-\<v, \gamma v\>_{\C^2} \im \sigma_3 \gamma v\).
\end{align*}
Again, since $ H^s$ is an algebra, we can bound each term by using the elementary inequality
\begin{align*}
\| hf\|_{H^s}\lesssim \(\|h\|_{L^\infty}+\|\partial_x h\|_{H^{s-1}}\)\|f\|_{H^s}.
\end{align*}
By a similar manner, we have the estimate for $\|[B,G]v\|_{ H^s}$.
\end{proof}

\begin{lemma}\label{lem:boundH1QW}
Let $T'>0$.
Suppose there exists $\delta_1>0$ such that for $\delta\in (0,\delta_1]$, 
$$\sup_{0\leq t\leq T'}\|v_\delta(t)\|_{ H^s}\leq M.$$
Then, there exists $\delta_0>0$ such that for $\delta\in(0,\delta_0]$, 
$$\sup_{0\leq t\leq T'}\|v_\delta(t)\|_{ H^{s+1}}\lesssim_{T', L,M} 1.$$
In particular, the implicit constant is independent of $\delta$.
\end{lemma}

\begin{proof}
For $0\leq \tau\leq \delta$, we set
\begin{align}\label{v1andv2}
v_{\delta,1}(\delta m+\tau):=v_\delta(\delta m+\tau, \delta m, \delta m).
\end{align}
Since we have
\begin{align*}
v_\delta(\delta m+\tau_1,\delta m+\tau_2,\delta m+\tau_3)= 
e^{-i\tau_3 A}
v_{\delta}(\delta m+\tau_1,\delta m+\tau_2,\delta m),
\end{align*}
 by the 3rd line of \eqref{def:v},
we see $$\|v_\delta(\delta m + \tau_1,\delta m +\tau_2,\delta m+\tau_3)\|_{ H^{s+1}}=\|v_{\delta}(\delta m+\tau_1,\delta m+\tau_2,\delta m)\|_{ H^{s+1}}.$$
Similarly, by the 2nd line of \eqref{def:v}, we have
$$v_{\delta}(\delta m+\tau_1,\delta m+\tau_2,\delta m)
=
v_{\delta,1}(\delta m+\tau_1)
-\im\int_0^{\tau_2} B 
v_\delta(\delta m + \tau_1, \delta m + \sigma, \delta m)
\,d\sigma.
$$
Thus,
\begin{align*}
\sup_{0\leq \tau_2\leq \tau}\|v_{\delta}(\delta m+\tau,\delta m+\tau_2,\delta m)\|_{H^{s+1}}\leq &\|v_{\delta,1}(\delta m+\tau)\|_{H^{s+1}}\\&+\widetilde C \tau \sup_{0\leq \tau_2\leq \tau}\|v_{\delta}(\delta m+\tau,\delta m+\tau_2,\delta m)\|_{H^{s+1}},
\end{align*}
where we have used assumption \eqref{Assumption,thm:1}.
Therefore, we conclude
\begin{align}\label{v2bound}
\|v_{\delta}(\delta m + \tau, \delta m+\tau, \delta m)\|_{ H^{s+1}}\leq (1+C \tau)\|v_{\delta,1}(\delta m+\tau)\|_{ H^{s+1}}\leq e^{C \tau} \|v_{\delta,1}(\delta m+\tau)\|_{ H^{s+1}}.
\end{align}

Now, since $v_{\delta,1}(\delta m+\tau)$
is the solution to
$
\im \partial_\tau v_{\delta,1}= G(v_{\delta,1})$ with $v_{\delta,1}(\delta m)=v_\delta(\delta m).
$
Therefore, since we can express $v_{1,\delta}=e^{\im g(\<v_\delta(\delta m),\gamma v_\delta(\delta m)\>)\gamma}v_\delta(\delta)$, we have 
\begin{align}\label{v1bound}
\|v_{1,\delta}(\delta m+\tau)\|_{ H^s}\lesssim_M 1.
\end{align}
Further, since
\begin{align*}
\left|\frac{d}{d\tau}\|v_{\delta,1}\|_{ H^{s+1}}^2\right|\leq\sum_{j=0}^{s+1}\sum_{k=0}^j2{}_jC_k\left |\<\partial_x^k \( g(\<v_{\delta,1},\gamma v_{\delta,1}\>_{\C^2})\)\gamma \partial_x^{j-k} v_{\delta,1},\partial_x^j v_{\delta,1}\>\right|,
\end{align*}
and by \eqref{v1bound}, we have 
\begin{align}\label{v1bound2}
\left|\frac{d}{d\tau}\|v_{\delta,1}\|_{ H^{s+1}}\right|\lesssim_M\|v_{\delta,1}\|_{ H^{s+1}}.
\end{align}
Therefore, by comparison theorem of ordinarily differential equation (or Gronwall's inequality), we have
\begin{align*}
\|v_{\delta ,1}(\delta m+\tau)\|_{ H^{s+1}}\leq e^{C_M \tau}\|v_{\delta}(\delta m)\|_{ H^{s+1}},
\end{align*}
where $C_M>0$ is the implicit constant in \eqref{v1bound2}.
Combining \eqref{v2bound} and \eqref{v1bound2}, we have
\begin{align*}
\|v_{\delta}(\delta m+\tau)\|_{H^{s+1}}\leq e^{C_0\tau}\|v_\delta(\delta m)\|_{H^{s+1}},
\end{align*}
with $C_0=C+C_M$.
Thus for $0\leq t\leq T'$, we have
\begin{align*}
\|v_\delta(t)\|_{ H^{s+1}}\leq e^{C_0 T'}\|u_0\|_{ H^{s+1}}.
\end{align*}
This gives us the conclusion.
\end{proof}

\begin{lemma}\label{lem:vboundimplyerror}
Let $T'>0$ and suppose 
\begin{align}\label{assHsbound}
\sup_{0\leq t\leq T'} \|v_\delta(t)\|_{H^s}\leq M.
\end{align}
Then, we have 
\begin{align}\label{boundtoconv}
\sup_{0\leq t\leq T'}\|v_\delta(t)-u_\delta(t)\|_{H^s}\lesssim_{T',L,M}\delta.
\end{align}
\end{lemma}

\begin{proof}
Set $w_\delta(t):=v_\delta(t)-u_\delta(t)$.
%
By \eqref{assHsbound} and Lemma \ref{lem:boundH1QW}, we have
\begin{align*}
\sup_{0\leq t\leq T'}\|v_\delta(t)\|_{ H^{s+1}}\lesssim_{T',L,M} 1.
\end{align*}
Next, by \eqref{def:v}, we have
\begin{align*}
\im \partial_t w_\delta&=\im\partial_{t_1} v_\delta+\im\partial_{t_2} v_\delta +\im\partial_{t_3}v_\delta-\im\partial_t u_\delta \\&=
\im\partial_{t_1}v_\delta - G(v_\delta) + \im \partial_{t_2}v_\delta - Bv_\delta +(A w_\delta +B w_\delta + G(v_\delta)-G(u_\delta)).
\end{align*}
Then, setting $F_{12}(t_1,t_2,t_3):=\im \partial_{t_1}v_\delta - G(v_\delta) + \im \partial_{t_2}v_\delta - Bv_\delta$, we have
\begin{align*}
\frac{d}{dt}\|w_\delta(t)\|_{ H^s}^2=2\Re\(\<\im w_\delta,F_{12}\>_{ H^s}+\<\im w_\delta, G(v_\delta)-G(u)\>_{ H^s}+\<\im w_\delta, B w_\delta\>_{ H^s}\).
\end{align*}
and thus
\begin{align}\label{diffofw}
\frac{d}{dt}\|w_\delta(t)\|_{ H^s}\leq\|F_{12}\|_{ H^s}+\|G(v_\delta)-G(u)\|_{ H^s}+\|B w_\delta\|_{ H^s}.
\end{align}
Recall we have $\|B w_\delta\|_{ H^s}\lesssim \|w_\delta\|_{ H^s}$.

By \eqref{def:v} we have $F_{12}(t,\delta m, \delta m)=0$.
Further, 
\begin{align*}
\im\partial_{t_3}F_{12} &= \im \partial_{t_1}(A v_\delta) - G'(v_\delta)(A v_\delta) + \im \partial_{t_2}(A v_\delta) - BAv_\delta \\&
= A F_{12} +[A, G](v_\delta) + [A,B]v_\delta.
\end{align*}
By lemma \ref{lem:com}, we have
\begin{align*}
\|[A,G]v_\delta\|_{ H^s}+\|[A,B](v_\delta)\|_{ H^s}\lesssim_{T',L,M} \|v_\delta\|_{ H^{s+1}}\lesssim_{T',L,M}1.
\end{align*}
Therefore, we have
\begin{align}\label{est:F12}
\|F_{12}(t,t,t)\|_{H^s}\leq \|F_{12}(t,t,\delta m)\|_{H^s}+ C \delta,
\end{align}
where $C=C_{T',L,M}>0$ is a constant.
Now, we set $$F_1(t_1,t_2):=F_{12}(t_1,t_2,\delta m)=\im \partial_{t_1}v_\delta - G(v_\delta).$$
By \eqref{def:v}, we have $F_1(t,\delta m)=0$ and 
\begin{align*}
\im \partial_{t_2}F_1=\partial_{t_1} ( B v_{\delta})+\im G'(v_\delta)( B v_\delta)= B F_1 + [B,G]v_\delta.
\end{align*}
The estimate of $\|F_1(t,t)\|_{H^s}$ need a little care since $B$ do not commutate with the derivatives.
Since, by Lemma \ref{lem:com}, we have
$
\|[B,G]v_\delta\|_{H^s}\lesssim_{T',L,M} 1,
$
we first get the estimate
\begin{align}\label{est:F1}
\|F_1(t,t)\|_{L^2}\lesssim_{T',L,M}  \delta.
\end{align}
Suppose that for $s'\leq s$, we have the estimate
\begin{align}\label{est:F13}
\|F_1(t,t)\|_{H^{s'-1}}\lesssim_{T',L,M}  \delta.
\end{align}
Since
\begin{align*}
\Re\<\partial_x^j F_1, -\im \partial_x^j\(B F_1\)\>=\sum_{k=0}^{j-1}{}_jC_k\Re\<\partial_x^j F_1, -\im \partial_x^{j-k}\mathbf s\cdot\bsig \partial_x^k F_1\),
\end{align*}
we have
\begin{align*}
\partial_t \|F_1(t,t)\|_{H^{s'}}\lesssim_{T',L,M} \delta + 1.
\end{align*}
Therefore, we obtain \eqref{est:F13} with $s'-1$ replaced by $s'$ and thus by induction we have \eqref{est:F13} with $s'-1$ replaced by $s$.
Therefore, substituting \eqref{est:F1} into \eqref{est:F12}, we have
\begin{align}\label{est:F122}
\|F_{12}\|_{ H^s}\lesssim_{T',L,M} \delta.
\end{align}

Next, since
\begin{align*}
G(v_\delta)-G(u_\delta)=\int_0^1 G'(u_\delta + \tau w_\delta) w_\delta \,d\tau. 
\end{align*}
we have
\begin{align}\label{est:G}
\|G(v_\delta)-G(u_\delta)\|_{ H^s}\lesssim_{T',L,M} \|w_\delta\|_{ H^s}.
\end{align}
Therefore, by \eqref{diffofw}, \eqref{est:F122} and \eqref{est:G}, we have
\begin{align}\label{est:wdelta}
\frac{d}{dt}\|w_\delta\|_{ H^s}\lesssim_{T',L,M} \delta+\|w_\delta\|_{H^{s}}.
\end{align}
This gives us the conclusion.
Indeed, if we have such inequality, setting $A(t)$ by
\begin{align*}
A(0)=\|w_\delta(0)\|_{H^s}=0,\quad A'(t)=\widetilde C(A+\delta),
\end{align*}
we have $\|w_\delta(t)\|_{H^s}\leq A(t)$, where $\widetilde C=\widetilde C_{T',L,M}>0$ is the implicit constant in \eqref{est:wdelta}.
Moreover, since we have $A(t)=e^{\widetilde Ct}(A(0)+\delta)-\delta$, we can conclude
\begin{align*}
\sup_{0\leq t\leq T'}\|w_\delta(t)\|_{H^s}\leq e^{\widetilde CT'}\delta,
\end{align*}
which is the desired estimate.
\end{proof}

\begin{proof}[Proof of Proposition \ref{prop:mainthm1}]
By Lemma \ref{lem:vboundimplyerror}, it suffices to prove \eqref{assHsbound} for $T'=T$.
Let $\widetilde C_{T,L}$ be the constant given by the assumption that $(T,L)$ satisfies $(\mathrm{Lip})_s$.
Without loss of generality, we can assume $\widetilde C_{T,L}L> \max(1,L)$.
Set $M=M_{T,L}:=4\widetilde C_{T,L}L$.
Let $C_{T,L,M}$ the implicit constant given in \eqref{boundtoconv} in Lemma \ref{lem:vboundimplyerror}.
Set $\delta_{T,L}:=C_{T,L,M_{T,L}}^{-1}$.
For $\delta\in (0, \delta_{T,L})$,  we set $$\mathcal T_\delta:=\{T'\in[0,T]\ |\ \sup_{0\leq t\leq T'}\|v_\delta(t)\|_{H^s}<M\}.$$
Then, it suffices to show $T\in \mathcal T_\delta$.

First, $0\in \mathcal T_\delta$ so $\mathcal T_\delta$ is not empty.
Further, since $v_\delta$ is continuous in $H^s$, we see that $\mathcal T_\delta$ is an open interval in $[0,T]$ (i.e.\ there exists an open interval $\mathcal O\subset \R$ s.t.\ $\mathcal T_\delta=[0,T]\cap \mathcal O$).
Now, suppose $T^*:=\sup \mathcal T_\delta<T$.
Then, for any $T'<T^*$, by Lemma \ref{lem:vboundimplyerror}, we have
\begin{align*}
\sup_{0\leq t\leq T'}\|v_\delta(t)\|_{H^s}&\leq \sup_{0\leq t\leq T'}\|u_\delta(t)\|_{H^s}+\sup_{0\leq t\leq T'}\|v_\delta(t)-u_\delta(t)\|_{H^s}\\&\leq \widetilde C_{T,L}L + C_{T,L,M}\delta\leq C_{T,L}L+1<\frac12 M.
\end{align*}
Therefore, by continuity, we have
\begin{align*}
\sup_{0\leq t\leq T^*}\|v_\delta(t)\|_{H^s}\leq \frac12 M,
\end{align*}
and thus for sufficiently small $\epsilon>0$, we have $T^*+\epsilon\in \mathcal T_\delta$, which contradicts the definition of $T^*$.
Therefore, we have $T\in \mathcal T_\delta$.
\end{proof}

\begin{proof}[Proof of Theorem \ref{thm:1}]
As mentioned in Remark \ref{rem:pr},
we need only bound the first term of \eqref{pthm2:2}.
We note that from Proposition \ref{prop:mainthm1}, Lemma \ref{lem:boundH1QW} and the assumption \eqref{Assumption,thm:1,2} we have
\begin{align}\label{boundednessUdelta}
\sup_{0\leq m\leq \lfloor T/\delta \rfloor}\|U_\delta(m)j_\delta u_0\|_{ H^{s+1}}\lesssim_{T,L} 1.
\end{align}
Further, notice that from the definition of $U_\delta$ and $\mathcal U_\delta$ and \eqref{i1}, we have 
\begin{align}\label{latticeeq}
\hid
\circ \ud(m) \circ \hid^{-1}\circ j_\delta u_0(x) = U_\delta(m) j_\delta u_0(x)\text{ for each }x\in \delta\Z.
\end{align}
Thus we can apply Proposition \ref{prop:shanon} and obtain
\begin{align}
\|\id \circ \ud(m) \circ \id^{-1}\circ j_\delta u_0 - U_\delta(m) j_\delta u_0\|_{ H^s}&\lesssim \delta\| U_\delta(m) j_\delta u_0\|_{H^{s+1}}  \lesssim_{T,L} \delta.
\end{align}
This completes the proof of Theorem \ref{thm:1}.
\end{proof}

\section*{Acknowledgments}  
M.M.\ was supported by the JSPS KAKENHI Grant Numbers 19K03579, JP17H02851 and\\ JP17H02853.
A. S. was supported by JSPS KAKENHI Grant Number JP26800054 and JP18K03327.

\medskip

Masaya Maeda

Department of Mathematics and Informatics,
Faculty of Science,
Chiba University,
Chiba 263-8522, Japan

{\it E-mail Address}: {\tt maeda@math.s.chiba-u.ac.jp}

\medskip

Akito Suzuki

Division of Mathematics and Physics,
Faculty of Engineering,
Shinshu University,
Nagano 380-8553, Japan

{\it E-mail Address}: {\tt akito@shinshu-u.ac.jp}

\end{document}